\documentclass[11pt]{CGC2}

\usepackage[english]{babel}
\usepackage{verbatim}
\usepackage{graphicx}

\newcommand{\C}{{\mathcal C}}

\newcommand{\dd}{\mathrm{d}}
\newcommand{\ww}{\mathrm{w}}
\newcommand{\cnk}{{\mathcal C}(n,k,q)}

\usepackage[english]{babel}
\usepackage{amsmath}
\usepackage{amssymb}
\usepackage{amsthm}
\usepackage[mathcal]{eucal}
\usepackage[latin1]{inputenc}

\begin{document}

\Logo{Paper submitted to Journal of Algebra and Its Applications (JAA)}

\begin{frontmatter}

\title{Computing the distance distribution of systematic non-linear codes}

{\author{Eleonora Guerrini}}
{{\tt (guerrini@posso.dm.unipi.it)}}\\
{{Department of Mathematics, University of Trento, Italy.}}

{\author{Emmanuela Orsini}}
{\tt{(orsini@posso.dm.unipi.it) }}\\
{{Department of Mathematics, University of Pisa, Italy.}}

{\author{Massimiliano Sala}} {\tt{(msala@bcri.ucc.ie)}}\\
{Boole Centre for Research in Informatics, UCC Cork, Ireland,\\
Department of Mathematics, University of Trento, Italy.}

\runauthor{E.~Guerrini, E.~Orsini, M.~Sala}

\begin{abstract}
The most important families of non-linear codes are systematic. A
brute-force check 
  is the only known method to compute their weight distribution and distance
distribution.
On the other hand, it outputs also all closest word pairs in the
code. In the black-box complexity model, the check is optimal
among closest-pair algorithms.
In this paper we provide a \Gr\ basis technique to compute the weight/distance 
distribution of any systematic non-linear code. Also our technique
outputs all closest pairs. Unlike the check, our
method can be extended to work on code families.
\end{abstract}

\end{frontmatter}

\section{Introduction}
In the celebrated paper \cite{CGC-cd-art-shannon} by Shannon, the mathematical
foundation of coding theory was laid.
Codes presented in that paper are non-linear, notably including  codes
used in the proof of the landmark Capacity theorem.
Although no proof  for linear codes of the Capacity theorem was known until the 60's
(\cite{CGC-cd-phdthesis-gallager}), coding theorists have been studying only linear codes, with
a few exceptions (\cite{CGC-cd-art-preparata},\cite{CGC-cd-art-goethal}).
This is not surprising, since linear codes have a nice structure, easy
to study and leading to efficient implementations.
Still, it is well-known that some non-linear codes have a higher distance
(or a better distance distribution) that any linear code with the same 
parameters (\cite{CGC-cd-art-preparata}, \cite{CGC-cd-book-handbook}). 
This translates into a superior decoding performance.

A class of non-linear codes that has received some attention is composed
of systematic non-linear codes, since they are easier to encode.
Moreover, the best known non-linear codes are systematic or equivalent
to systematic ones (\cite{CGC-cd-art-goethal},\cite{CGC-cd-book-handbook},
\cite{CGC-cd-art-hammons}), so that no performance degradation
is shown while restricting to systematic codes.
In this paper our main result is a method, based
on  \Gr\ basis computations, that allows to determine the (distance
and) distance distribution of a systematic code. 
No other method is known, except for the ``brute-force'' approach consisting
of checking the mutual distance of any pair of codewords.
Both our method and the brute-force approach output the {\em closest pairs},
i.e. the codeword pairs whose distance is minimal. 
We show that the complexity of the ``brute-force'' approach matches 
the complexity of the closest pair problem and so this method is optimal 
(the proof is given within the black-box complexity model with distance oracle).

\section{Notation and preliminary results}
\label{prel}
\newcommand{\hi}{\hspace{3cm}}
\newcommand{\h}{\hspace{0.4cm}}

Let  $m\geq 1$ be a natural number.
Let $\KK$ be a field, $\KKK$ be the algebraic closure of $\KK$
and $I$ be an ideal in the polynomial ring $\KK[Y]=\KK[y_1, \dots, y_m]$.
For any $q\geq 1$, we denote by $E_q[Y]\subset \KK[Y]$ the   set of polynomials  
$
E_q[Y]=\{y_1^q-y_1,\ldots,y_m^q-y_m  \} \,.
$
Given a polynomial $f \in \KK[Y]$, we denote by $\mathcal{V}(f)$
the set of all zeros of $f$ in $(\KKK)^m$.
Given an ideal $I \subseteq \KK[Y]$, we denote by $\mathcal{V}(I)$
the set of all zeros of $I  $ in $  (\KKK)^m$. 
Let $S \subset  (\KKK)^m$. The set of all polynomials $f \in \KK[Y]$
such that $f(a_1,\dots,a_m)=0$ for any point $(a_1,\dots,a_m) $ in $S$
forms an ideal in polynomial ring $\KK[Y]$, called the {\bf{vanishing ideal}} of $S$  
and   denoted by $\mathcal{I}(S)$.
If $L\subset \KK[Y]$, we denote by $\langle L\rangle$ the ideal in $\KK[Y]$ generated by $L$.

Let $\FF_q$ be the finite field with $q$ elements and 
$(\FF_q)^m$ be the natural $m$--dimensional vector space over $\FF_q$.

\begin{definition}
 Let $1\leq t\leq m$. We denote by $\mathcal{M}_{m,t,q}$ the following
  set:
$$
  \mathcal{M}_{m,t,q} = \{ y_{h_1}\cdots y_{h_t} \, 
    \mid   1 \leq h_1< \ldots < h_t\leq m \}.   
$$
\end{definition}
From now  on, we will use  $\mathcal{M}_{m,t}$ instead of  $\mathcal{M}_{m,t,q}$.
We will also shorten $y_{h_1}\cdots y_{h_t}$ to $Y_L$, where $L=\{h_1,\ldots,h_t\}$.

Let $s$ be an integer $1\leq s \leq m-1$. 
We fix in
$\mathbb{F}_q[y_1,\ldots ,y_s,t_1,\ldots ,t_{m-s}]=\FF_q[Y,T]$, the
lexicographic order   $y_1< y_2< \ldots < y_s <t_1<\ldots <t_{m-s} $.
Let $I$ be an ideal in  $\FF_q[Y,T]$ we denote by
$G(I) \subset \mathbb{F}_q[Y,T]$  the minimal reduced \GR \  basis of
$I$ w.r.t.  $<$ ordering (\cite{CGC-alg-phdthesis-buchberger,CGC-alg-art-thesisbuchbergerEN,CGC-alg-book-cox1}).

Let
$\phi:(\mathbb{F}_q)^k \rightarrow (\mathbb{F}_q)^n$
be an injective function and let $C$ be ${\rm Im}(\phi)$.
We say that $C$ is an $(n,k,q)$ {\bf code}.
Any $ c\in C$  is called a  {\bf word}.
Let  $\pi:(\mathbb{F}_q)^n \rightarrow (\mathbb{F}_q)^k $ be
$\pi(a_1,\ldots ,a_n)= (a_1,\ldots,a_k)$.
We say that $C$ is  {\bf systematic } if
$ (\pi \circ\phi)(v)=v$ for any $v\in (\FF_q)^k$.
We denote by $\mathcal{C}(n,k,q)$ the class of systematic  $(n,k,q)$  codes.

For any two vectors  $v_1,v_2 \in (\FF_q)^n$,  $\dd(v_1,v_2)$ denotes
the {\bf (Hamming) distance}     between
$v_1$ and $v_2$.
For any $v \in(\FF_q)^n  $, $\ww(v)$ denotes the {\bf weight } of
$v$. Let   
  $C \in \mathcal{C}(n,k,q)$,   $\dd(C)$ denotes the     distance   of
  $C$,    
  $B_i=B_i(C)$ denotes the number of codewords in $C$ 
  with weight $i$.
  Integer set $\{B_0,B_1, \dots, B_n\}$ is called the 
  {\bf{weight distribution}} of $C$ and
  $A_i=A_i(C)$ denotes the number 
of (unordered) codeword pairs  with distance $i$.
Integer set $\{A_1, \dots, A_n\}$ is called the 
{\bf{distance distribution}} of $C$.
 
A code $C$ is  {\bf{distance-invariant}} if for 
any $1\leq i\leq n$ and any $c,c'\in C$, 
$$
  |\{y \in C \mid  \dd(c,y)=i\}| \; =\; |\{y \in C \mid  \dd(c',y)=i\}| \;.
$$
Clearly linear codes are distance-invariant.
The distance distribution of distance-invariant codes 
(containing the zero vector) can be immediately obtained 
from their weight distribution.
Note that many optimal codes are distance-invariant codes (\cite{CGC-cd-art-hammons}).


\label{weights}
We want to show some relations between some
sets of vectors in $(\FF_q)^n$ having a weight property and the \GR\ bases of associated
ideals.
We   use polynomial ring $\FF_q[Y]$  with any   term-order, since our results here 
hold for any (admissible) term-order.\\
The {\it{elementary symmetric functions}}  are the polynomials
$\sigma_1 = y_1+\cdots+y_m,\ldots,
\sigma_{m} = y_1 y_2 y_3 \cdots y_{m-2}y_{m-1}y_m$.
\begin{definition}
Let $t \in \NN$ be s.t. $1\leq t \leq m$. 
We denote by  $I_{m,t}$ the ideal:
$$
 I_{m,t}= \langle \{\sigma_t,\dots,\sigma_m\} \cup E_q[Y]\rangle\subset \FF_q[Y] \,.
$$
\end{definition}
Our aim is to determine the reduced \GR\ basis of $I_{m,t}$.
For this, we need a preliminary result.
\begin{lemma}\label{mezzo_ideale}
Let  $l \in \NN$ be such that $1 \leq l\leq m-1$. Then 
$$
{\mathcal M}_{m,l} \subset \langle  {\mathcal M}_{m,l+1} \cup \{\sigma_l\} \cup  E_q[Y] \rangle \,.
$$
\end{lemma}
\begin{proof}
Let $I$ be the ideal $I= \langle  {\mathcal M}_{m,l+1}\cup \{\sigma_l\} \cup  E_q[Y]\rangle$.\\
Let ${\mathcal A}= \{ A \subset \{1,..,m \} \mid |A|=l \} $ 
and $L=\{i_1,\ldots,i_l\} \in {\mathcal A}$. We have:
\begin{equation}
    \label{eq:sigma_l}
    \sigma_l=\sum_{ A \in {\mathcal A}\atop{\{i_1\} \subset A
}} Y_A + \sum_{ B\in {\mathcal A}\atop{\{i_1\} \not\subset B}}Y_B 
\,\,\,  y_{i_1}^{q-1}\sigma_l=y_{i_1}^{q-1}\sum_{ A\in {\mathcal A}\atop{\{i_1\} \subset A
}}Y_A +y_{i_1}^{q-1} \sum_{B \in {\mathcal A}\atop{ \{i_1\}\not\subset B}}Y_B 
\,.
  \end{equation}
We note that $y_{i_1} \sum_{\{i_1\}\not\subset B}Y_B $ is a sum of monomials
 in ${\mathcal M_{m,l+1}}$,
 hence\\
$y_{i_1}^{q-2} ( y_{i_1} \sum_{\{i_1\}\not\subset B}Y_B) = 
y_{i_1}^{q-1} \sum_{\{i_1\}\not\subset B}Y_B $ is in $I$.
By (\ref{eq:sigma_l}), since $y_{i_1}^{q-1}\sigma_l \in I$, we obtain that
$$
y_{i_1}^{q-1}\sum_{A \in {\mathcal A}\atop{ \{i_1\} \subset A}}Y_A \;\in\; I\,.
$$
From (\ref{eq:sigma_l}), since any monomial in $\sum_{\{i_1\} \subset A} Y_A$
contains $y_{i_1}$, by reduction w.r.t. $E_q[Y]$ we have
$
\sum_{A\in {\mathcal A} \atop{\{i_1\} \subset A}}Y_A \;\in\; I. 
$
We can write:
 \begin{equation}
\label{eq:sigma_l_y_iy_j}
\sum_{A \in {\mathcal A} \atop{ \{i_1\} \subset A}}Y_A=\sum_{B \in {\mathcal A} \atop{ \{i_1,i_2\} \subset B}}Y_B+
\sum_{C \in {\mathcal A } \atop{ \{i_1,i_2\} \not\subset C, \{i_1\} \subset C }}Y_C \,,
 \end{equation}
and from (\ref{eq:sigma_l_y_iy_j}) we get 
$$
   y_{i_2}^{q-1}\sum_{A \in {\mathcal A}\atop{\{i_1\} \subset
       A}}Y_A=y_{i_2}^{q-1}\sum_{B \in {\mathcal A}\atop{\{i_1,i_2\}
       \subset B}}
Y_B+y_{i_2}^{q-1}
  \sum_{\{i_1,i_2\} \not\subset 
  C\atop{\{i_1\} \subset C,C \in {\mathcal A}}} Y_C \,.
$$
Since 
$y_{i_2}^{q-1}
\sum_{\{i_1,i_2\} \not\subset B\atop{\{i_1\} \subset B}}Y_B =y_{i_2}^{q-2}(y_{i_2}
\sum_{\{i_1,i_2\} \not\subset B\atop{\{i_1\} \subset B}}Y_B)$
and $y_{i_2}\sum_{\{i_1,i_2\} \not\subset B\atop{\{i_1\} \subset B}}Y_B$ is in $I$, 
then the former  is in $I$.\\
Similarly, $y_{i_2}^{q-1}\sum_{\{i_1,i_2\} \subset B}Y_B$ is in $I$, 
and by reduction w.r.t. $E_q[Y]$ we have
   \begin{equation}
     \label{eq:ultima}
     \sum_{B\in{\mathcal A} \atop{\{i_1,1_2\} \subset B}}Y_B  \;\in\; I \,.
   \end{equation}
In the same way, restarting from (\ref{eq:ultima}), 
we will eventually deduce that
 $
 \sum_{D\in {\mathcal A}\atop{L \subset D }}Y_D   
 $ is in $I$.
Since $|D|=|L|$, then $\sum_{L \subset D} Y_D =Y_L=y_{i_1}\cdots
y_{i_l}$, i.e. it is an element of $ {\mathcal M}_{m,l}$. But $L$ is a generic element of
 ${\mathcal A}$ and  $ {\mathcal M}_{m,l}=\{  Y_L \}_{L\in {\mathcal A} }$.
\end{proof}
 We now determine the reduced \GR\ basis  for $I_{m,t}$.

\begin{theorem}\label{baseidealepesi}
Let  $t\in \NN$ be such that $1\leq t\leq m$. 
Let  $G(I_{m,t})$  be the reduced \GR \ basis of $I_{m,t}$.  
Then:
 \begin{align*}
   G(I_{m,t})&=E_q[Y] \cup {\mathcal M}_{m,t}, & \mbox{ for } t \geq 2 \,,\\
G(I_{m,t})&=\{ y_1,..,,y_m\},                  & \mbox{ for } t =1     \,.
 \end{align*}
\end{theorem}
\begin{proof}
  The statement is easily proved by checking Buchberger's criterion,
  that is, that all $S-$polynomials coming from $G(I_{m,t})$ are reduced to
  zero by reduction via $G(I_{m,t})$.
\end{proof}

%
An obvious consequence of Theorem \ref{baseidealepesi} is that
${\mathcal M}_{m,t}\subset I_{m,t}.
$
For any $1\leq i\leq m$, we denote by $P_i$ the set $P_i=\{ c \in (\FF_q)^m \ | \ \ww(c)=i \ \} $,
and by $Q_i$ the set $Q_i=\sqcup_{0\leq j\leq i} P_j$.
Set $P_i$ contains all vectors of weight $i$ and set $Q_i$ contains all vectors of
weight up to $i$.
In the next theorem we describe the \Gr\ basis of the vanishing ideal of $Q_i$.
\begin{theorem}\label{idealepesi}
Let $t$ be an integer such that $0\leq t\leq m-1$. 
Then 
$$
  {\mathcal I}(Q_t) \;=\; \langle\{\sigma_{t+1},\ldots,\sigma_m\}\cup E_q[Y]\rangle \;=\;I_{m,t+1} \;,
$$
hence the reduced \Gr\ basis $G$ of ${\mathcal I}(Q_t)$ is
\begin{align*}
   G & =E_q[Y] \cup {\mathcal M}_{m,t}, & \mbox{ for } t \geq 1 \,,\\
   G & =\{ y_1,\dots,y_m\},        & \mbox{ for } t =0     \,.
 \end{align*}
\end{theorem}
\begin{proof}
It is enough to show both inclusions
 $$
   \langle  \{ \sigma_{t+1},\ldots,\sigma_m\}\cup E_q[Y] \rangle
   \subseteq {\mathcal I}(Q_t),\,\quad
  {\mathcal V}\big( \langle \{\sigma_{t+1},\ldots,\sigma_m\} \cup E_q[Y]\rangle\big)\subseteq 
  {\mathcal V}({\mathcal I}(Q_t)) \,.
$$
\indent
  For any $c$ in $Q_t$, we have $c \in (\FF_q)^m$ and 
  $\sigma_{t+1}(c)=0,\ldots,\sigma_m(c)=0$,
  hence $ \langle \{\sigma_{t+1},\ldots,\sigma_m\} \cup E_q[Y]\rangle \subseteq {\mathcal I}(Q_t)$.  

Since ${\mathcal V}({\mathcal I}(Q_t))=Q_t$, we need to show
  $$
  a\in {\mathcal V}(\langle \{ \sigma_{t+1},\ldots,\sigma_m\}\cup E_q[Y] \rangle) 
  \implies a\in Q_t \,. 
  $$
  We use the relation  $(\FF_q)^m= P_0 \sqcup P_1  \sqcup \ldots \sqcup P_m$ and we observe
  that \\${\mathcal V}(\langle \{\sigma_{t+1},\ldots,\sigma_m\} \cup E_q[Y]\rangle)\neq \emptyset$, 
  since it contains the zero vector.\\
  Suppose by contradiction that there is an $a \in P_{t+1} \sqcup P_{t+2}\sqcup\ldots\sqcup P_m$ 
  such that $\sigma_{t+1}(a)=0,\ldots,\sigma_m(a)=0$. 
  Since $a$ has weight $r$, with  $r\geq t+1$,
 it follows that there exists a monomial $\overline{\sf m}\in {\mathcal M}_{m,r}$ such that 
 $\overline{\sf m}(a)\neq 0$ and ${\sf m}(a)=0$ for any other monomial ${\sf m}$ in ${\mathcal M}_{m,r}$. 
 Hence $\sigma_r(a)\neq 0$, contradicting the hypothesis 
  $ a\in {\mathcal V}(\langle  \{\sigma_{t+1},\ldots,\sigma_m\} \cup E_q[Y] \rangle) $.
\end{proof}


\section{A distance-computing algorithm}
\label{algoMEGA}

In this section we propose a computational method to find the weight distribution of a code $C$ in $\cnk$.
We also extend this method to find the distance and the distance distribution of $C$.
From now on, $t$ will be understood to satisfy $1\leq t\leq n$ and $C$
will denote a code in $\cnk$.

\subsection{\Gr \  basis of a non-linear systematic code}
We apply our previous results to give a structure for the \GR\, basis of
$C$.
Let $n,k \in \mathbb{N}$ and 
 $\FF_q[X,Z]$ be as in Section \ref{prel}.
We also use $E_q[X]\subset \FF_q[X,Z]$ with the obvious meaning.

We can view $C$
as a set of points in $(\FF_q)^n\subset
({\overline{\FF}}_q)^n$ and hence as a $0$-dimensional variety, so that ${\mathcal I}(C)$ is 
its vanishing ideal in $\mathbb{F}_q[X,Z]$. 
We describe the reduced \GR\ basis of ${\mathcal I}(C)$ w.r.t. lex. 
\begin{theorem}
\label{base}
 Let  G be the reduced
 \GR\,  basis for $I={\mathcal I}(C)$ w.r.t. lex with 
$x_1 < \dots < x_k < z_1 < \dots <z_{n-k}$.
Then:
$$G=E_q[X] \cup \{ z_1-{\sf f}_1,\ldots,z_{n-k}-{\sf f}_{n-k}\} 
$$
for some ${\sf f}_j \in \FF_q[X]$, $1 \leq j \leq  n-k$.
\end{theorem}
\begin{proof}
   Since $C \subset (\FF_q)^n$ and $C$ is systematic,
   $E_q[X]\subset G$. The existence of polynomials $z_i-{\sf f}_i$
 follows from the fact that any non-systematic component
   depends only on the $X$ block of variables.
\end{proof}
\begin{example}
Let $C$ be the following $(4,2,2)$ code: 
$$
  C=\{(0,0,0,1),(0,1,0,1),(1,0,0,1),(1,1,0,0)\} \,.
$$
The reduced \Gr \  basis of $\mathcal{I}(C) \subset \FF_2[x_1,x_2,z_1,z_2]$ 
w.r.t. the lex 
ordering with $x_1 <x_2<z_1<z_2$ is
$
G(C)=\{x_1^2+x_1,x_2^2+x_2,z_1,z_2+x_1x_2+1\} \,.
$
So ${\sf f}_1=0$ and ${\sf f}_2=x_1x_2+1$.
\end{example}
 
Any set of polynomials endowed with the structure of  Theorem \ref{base} is a
\GR\, basis for the ideal generated by itself, and such ideal is
zero-dimensional.  Hence, the corresponding variety is finite and
satisfies the properties of the codes in $\cnk$.  In this sense we
can identify polynomial sets (with the structure of the theorem) and codes in $\cnk$.  
We make it explicit in the following theorem.

\begin{theorem}
Let ${\mathcal A}_{k,n} $ be the set 
$$
{\mathcal A}_{k,n}=\{ ({\sf f}_1,\ldots,{\sf f}_{n-k}) \mid {\sf f}_j: (\mathbb{F}_q)^k  \longrightarrow \mathbb{F}_q,\, 1\leq j\leq n-k\}.
$$ 
There is a bijection ${\mathcal A}_{k,n} \leftrightarrow \cnk$ given by 
$$
  ({\sf f}_1,\ldots,{\sf f}_{n-k}) \longleftrightarrow  
  G= E_q[X]\cup \{ z_1-{\sf f}_1,\ldots, z_{n-k}-{\sf f}_{n-k}\} \,.
$$
\end{theorem}
\subsection{Weight distribution  for non-linear systematic codes}

The first computational method we propose is a method to obtain the
weight distribution for $C$.
\begin{definition}
  Let $G(C)$ be the reduced \GR\  basis for $C$, 
  $G(C)= E_q[X] \cup \{z_1-{\sf f}_1(X),\dots, z_{n-k}-{\sf f}_{n-k}(X)  \}$.
  We denote by ${\mathcal W}_C^t$ the following ideal in $\FF_q[x_1,\ldots,x_k]$:
  $$
  {\mathcal W}_C^t=\langle E_q[X] \cup 
  \{ {\sf m}\bigl( x_1,\ldots,x_k,{\sf f}_1(X),\dots,{\sf f}_{n-k}(X) \bigr) \mid  
            {\sf m} \in {\mathcal M}_{n,t} \} \,.
  $$
\end{definition}
\begin{lemma}\label{pisa}
$
{\mathcal V}(  {\mathcal W}_C^t  )\neq \emptyset \iff \exists c\in C \ \ {\mbox  s.t. }\ \  \ww(c)\leq t-1. 
$
\end{lemma}
\begin{proof}
  Let $a \in {\mathcal V}( {\mathcal W}_C^t)$,  $a=(a_1,\ldots,a_k)$.
 We have that $c$ is in $C$, where\\
  $c=(a_1,\ldots,a_k,{\sf f}_1(a),\ldots,  {\sf f}_{n-k}(a))$.
 Since $a \in {\mathcal V}( {\mathcal W}_C^t)$, we have
 $$
   {\sf m}(a)={\sf m}\bigl( a_1,\ldots,a_k,{\sf f}_1(a),\ldots,{\sf f}_{n-k}(a) \bigr)= 0 
   \quad \forall {\sf m} \in {\mathcal M}_{n,t} \,.
 $$
Which means
$c \in {\mathcal V}(I_{m,t})$, hence $c \in Q_{t-1}$ (Theorem \ref{idealepesi}), 
i.e. $\ww(c)\leq t-1$.

The converse can be easily proved by reversing our previous argument.
\end{proof}
Lemma \ref{pisa}  shows that a point in ${\mathcal V}(  {\mathcal W}_C^t  )$
matches a codeword $c$ in $C$ with $\ww(c)\leq t-1$, from which we can easily  derive the 
main result of this subsection.
\begin{theorem}
$
  B_{t-1}=|{\mathcal V}( {\mathcal W}_C^t)|\setminus |{\mathcal V}( {\mathcal W}_C^{t-1})| \,.
$
\end{theorem}

If $C$ is a distance-invariant code,  from the weight distribution of $C$ we can immediately get
its distance distribution. 

\subsection{Distance and distance distribution of non-linear systematic codes}
We now  propose a computational method to find the distance
 of $C$. 

\begin{definition}
Let ${\sf f}_1, \dots, {\sf f}_{n-k}$ be
as
in Theorem \ref{base}.
We denote by $\FF_q[X,\tilde X]$ the polynomial ring
$\FF_q[x_1,x_2,\ldots,x_k,\tilde{x}_1,\tilde{x}_2,\ldots,\tilde{x}_k]$.
In the polynomial module $\left( \FF_q[X,\tilde X] \right)^n$, we denote
by $L_{n,k,t}$ the   polynomial vector:
$$
  L_{n,k,t}= \left(x_1-\tilde{x}_1, \ldots, x_k-\tilde{x}_k,   
              {\sf f}_1(X)- {\sf f}_1(\tilde X), \ldots,
            {\sf f}_{n-k}(X)- {\sf f}_{n-k}(\tilde X) \right).
$$
\end{definition}
\begin{definition}\label{def_ict}
Let  ${\sf f}_1, \dots, {\sf f}_{n-k}$ be as
in Theorem \ref{base}.
We denote by 
${\mathcal I}_C^t$ the ideal in $\FF_q[X, \tilde X]$ generated by :
$$
  \{ x_i^q-x_i,\tilde{x_i}^q-\tilde{x_i} \mid 1\leq i\leq k \} \cup
  \{ {\sf m} (L_{n,k,t} ) \mid 
      {\sf m} \in\mathcal{M}_{n,t} \} \,.
$$
\end{definition}
In $(\FF_q)^k× (\FF_q)^k$ we denote by $\Delta_k $ the 
diagonal, i.e. the set of points
$a=(a_1,\ldots,a_k,\tilde{a_1},\ldots,\tilde{a_k})$
 such that $a_i=\tilde{a}_i$, $ 1\leq i\leq k$.
 
We need to take into account the diagonal because, clearly,
$\Delta_k \subset {\mathcal    V}({\mathcal I}_C^t)$ (note that $\Delta_k={\mathcal V}({\mathcal I}_C^1)$).
\begin{theorem}\label{mega1}
$
{\mathcal V}({\mathcal I}_C^t)\neq \Delta_k\iff \exists\, c_1,\,
c_2\in C \,{\mbox{such that}}\,\,  \dd(c_1,c_2)\leq t-1 \,.
$
\end{theorem}
\begin{proof}
Let $(a,{\tilde a}) \in (\FF_q)^k\times (\FF_q)^k \setminus \Delta_k$.
We take two codewords  $c,\,  {\tilde c}$ such that: 
$$
c=(a_1,...,a_k,{\sf f}_1(a),\ldots,{\sf f}_{n-k}(a)),\quad
{\tilde c}=({\tilde a}_1,\ldots,{\tilde a}_k,{\sf f}_1({\tilde a}),\ldots,{\sf f}_{n-k}({\tilde a}))\,.
$$
From this, we obtain that
$$
  (a,{\tilde a}) \in {\mathcal V}({\mathcal I}_C^t)\neq \Delta_k 
  \iff {\sf m}(c-{\tilde c})=0, \quad
  c\neq {\tilde c} ,\, {\sf m} \in {\mathcal M}_{n,t}\; ,
$$
since $c-{\tilde c}= (a_1-{\tilde a}_1,...,a_k-{\tilde a}_k,
{\sf f}_1(a)-{\sf f}_1({\tilde a}),...,{\sf f}_{n-k}(a)-{\sf f}_{n-k}({\tilde a}))$.
But ${\sf m}(c-{\tilde c})=0 $ $\forall {\sf m} \in {\mathcal M}_{n,t} $
means $c-{\tilde c}\in {\mathcal V}(I_{n,t})$, which is
equivalent to $\ww(c-{\tilde c})\leq t-1 $ (Lemma \ref{pisa}).
\end{proof}
We can easily derive  the following result.
\begin{corollary}\label{mega}
$
{\mathcal V}({\mathcal I}_C^t)= \Delta_k\iff \dd(C)\geq t.
$
\end{corollary}

From Corollary \ref{mega}, an algorithm is directly designed to compute
the distance of   $C $.

\begin{center}

\newenvironment{tab}[1]%
{\begin{tabular}{|#1|}\hline}%
{\hline\end{tabular}}

\begin{tab}{l}



$\,\,$ $j=2$\\

\noindent

{\bf While} \,\, ${\mathcal V}({\mathcal I}_C^j) =\Delta_k $ \,\,\,{\bf do}\\

$\, \, \, $ $ j:=j+1$;\\

\noindent
{\bf Output}  $j-1$\\
\end{tab}
\end{center}
\begin{example}
  Let $C=\{[0, 0, 2, 0], [0, 1, 0, 0]$, $[0, 2, 0, 2], [1, 0, 2, 2]$, $[1, 1, 2, 1]$, $[1, 2, 1, 2], $[2, 0, 1, 0], [2, 1, 1, 1]$, [2, 2, 0, 1]\}$
be a $(4,2,3)$ code. 
The \Gr \ basis $G(C)$ $\in$ $\FF_3[x_1,x_2,z_1,z_2]$ w.r.t. lex
$x_1 < x_2 <z_1<z_2$ is
$$
G(C)=\{x_1^3 - x_1, x_2^3 - x_2, z_1 - x_2^2 + x_1^2x_2 - x_1^2 + x_1 + 1, z_2 - x_2^2 + x_1x_2 + x_2 - x_1^2 - x_1\}.$$
So ${\sf f}_1= x_2^2 - x_1^2x_2 + x_1^2 - x_1 - 1$ and 
${\sf f}_2= x_2^2 - x_1x_2 - x_2 + x_1^2 + x_1$. We choose as monomial
order in $\FF_3[x_1,x_2,\tilde{x}_1,\tilde{x}_2]$ the degrevlex order
with $x_1>x_2>\tilde{x}_1>\tilde{x}_2$.
Computing the  \GR\, basis of $\mathcal{I}_C^2$, we have  
 $G(\mathcal{I}_C^2)=\{\tilde{x}_1 - x_1, \tilde{x}_2 - x_2, x_1^3 - x_1, x_2^3 - x_2,\tilde{x}_1^3 - x_1, \tilde{x}_2^3- x_2\}$.
Since $\tilde{x}_1 - x_1, \tilde{x}_2 - x_2\in G$, then  
$\mathcal{V}(\mathcal{I}_C^2)$ is the diagonal, so  $d \geq 2$. 
Computing then
$ G(\mathcal{I}_C^3)$, we find  
\small{$$G(\mathcal{I}_C^3)=\{ \tilde{x}_2^2-\tilde{x}_2,
 \tilde{x}_1^2-\tilde{x}_1,
 x_2^2-x_2,
 x_1^2-x_1,
 x_1\tilde{x}_1\tilde{x}_2-x_1\tilde{x}_2,
 x_1x_2\tilde{x}_2-x_2\tilde{x}_1\tilde{x}_2,
 x_1x_2\tilde{x}_1-x_2\tilde{x}_1\}\,.
$$}
This time $\tilde{x}_1-x_1\not\in G(\mathcal{I}_C^3) $  and the
variety $\mathcal{V}(\mathcal{I}_C^3)  \neq \Delta_3$, which implies $d=2$.
\end{example}

Theorem \ref{mega1} suggests a way to compute the distance
distribution of   $C$. 
Indeed, we showed that a point in ${\mathcal V}(\mathcal{I}_C^t)$ is a
pair of codewords with  distance less than $t$.
From that, we have the following.
\begin{corollary}
$
{\mathcal V}({\mathcal I}_C^t)=\{(c_1,c_2) \mid c_1\neq c_2 \in C, \ \dd(c_1,c_2)\leq t-1 \} 
\cup \Delta_k \,.
$
\end{corollary}
Let $c_1, c_2 \in C$ be such that $\dd(c_1,c_2)= i $, for $1\leq i \leq t-1$. 
In ${\mathcal V}(\mathcal{I}_C^t)$ there is both the point that matches the
pair of words $(c_1,c_2)$ and the point that matches the pair $(c_2,c_1 )$.
The following fact is then obvious.
\begin{fact} \label{fact1}
$$
 \sum_{1\leq i\leq t-1} A_i \;=\; \frac{|{\mathcal V}( \mathcal{I}_C^t)| \setminus
|\Delta_k| }{2}  \;.
$$
\end{fact}
\begin{example}
\label{distr_dist}
Let $C=\{[0, 0, 0, 0]$, $[0, 1, 0, 0], [0, 2, 0, 0]$, $[1, 0, 0, 0], [1, 1, 0, 2]$, $[1, 2, 0, 2], [2, 0, 0, 2]$, $[2, 1, 0, 0], [2, 2, 0, 0]\}$. Clearly,
$C$ is in $\C(4,2,3)$. 
The distance distribution of $C$ can be determined by hand: $A_1=8,\,A_2=20,\,A_3=8,\,A_4=0$.
We want to compute all pairs of words $(c_1,c_2)$ with $\dd(c_1,c_2) \leq 2$. \\
To accomplish this, we start from the input basis of ideal $\mathcal{I}_C^3$:
\begin{itemize}
\item[] \small{$\mathcal{I}_C^3=\langle
x_2^2\tilde{x}_1^2\tilde{x}_2 - x_1^2\tilde{x}_1^2\tilde{x}_2 - x_1\tilde{x}_1^2\tilde{x}_2 - \tilde{x}_1^2\tilde{x}_2 - x_1^2x_2^2\tilde{x}_1\tilde{x}_2 + x_2^2\tilde{x}_1\tilde{x}_2 + x_1^2\tilde{x}_1\tilde{x}_2 - \tilde{x}_1\tilde{x}_2 - x_1^2x_2^2\tilde{x}_2 - x_1^2\tilde{x}_2 + x_1\tilde{x}_2 + x_1^2x_2\tilde{x}_1^2+ x_1x_2\tilde{x}_1^2- x_1^2x_2 - x_1x_2,
x_1^2x_2\tilde{x}_1^2\tilde{x}_2 + x_1x_2\tilde{x}_1^2\tilde{x}_2 - x_1^2x_2\tilde{x}_2 - x_1x_2\tilde{x}_2 - x_1^2x_2^2\tilde{x}_1^2- x_1x_2^2\tilde{x}_1^2+ x_1^2x_2^2+ x_1x_2^2, \,\,\,\,
\tilde{x}_2^3 - \tilde{x}_2, \,\,\,
x_2^3 - x_2, \,\,\,\,
\tilde{x}_1^3 - \tilde{x}_1, 
x_1^3 - x_1,  
x_2\tilde{x}_1^2\tilde{x}_2^2+ x_2\tilde{x}_1\tilde{x}_2^2- x_1^2x_2\tilde{x}_2^2- x_1x_2\tilde{x}_2^2- x_1^2\tilde{x}_1^2\tilde{x}_2 - x_1\tilde{x}_1^2\tilde{x}_2 - \tilde{x}_1^2\tilde{x}_2 - x_1^2x_2^2\tilde{x}_1\tilde{x}_2 + x_1^2\tilde{x}_1\tilde{x}_2 - \tilde{x}_1\tilde{x}_2 + x_1x_2^2\tilde{x}_2 - x_1^2\tilde{x}_2 + x_1\tilde{x}_2 + x_1^2x_2\tilde{x}_1^2+ x_1x_2\tilde{x}_1^2- x_1^2x_2 - x_1x_2,
x_1\tilde{x}_1^2\tilde{x}_2^2- \tilde{x}_1^2\tilde{x}_2^2+ x_2^2\tilde{x}_1\tilde{x}_2^2+ x_1^2\tilde{x}_1\tilde{x}_2^2- \tilde{x}_1\tilde{x}_2^2- x_1x_2^2\tilde{x}_2^2+ x_1^2\tilde{x}_2^2- x_1\tilde{x}_2^2- x_2\tilde{x}_1^2\tilde{x}_2 - x_2\tilde{x}_1\tilde{x}_2 + x_1^2x_2\tilde{x}_2 + x_1x_2\tilde{x}_2 - x_1x_2^2\tilde{x}_1^2- x_2^2\tilde{x}_1^2- x_1^2x_2^2\tilde{x}_1 + x_2^2\tilde{x}_1 + x_1^2x_2^2+ x_1x_2^2,
x_1x_2\tilde{x}_1\tilde{x}_2^2+ x_2\tilde{x}_1\tilde{x}_2^2- x_1^2x_2\tilde{x}_2^2- x_1x_2\tilde{x}_2^2- x_1^2\tilde{x}_1^2\tilde{x}_2 + \tilde{x}_1^2\tilde{x}_2 - x_1^2x_2^2\tilde{x}_1\tilde{x}_2 - x_1x_2^2\tilde{x}_1\tilde{x}_2 - x_1^2\tilde{x}_1\tilde{x}_2 + \tilde{x}_1\tilde{x}_2 + x_1^2x_2^2\tilde{x}_2 + x_1x_2^2\tilde{x}_2 + x_1^2x_2\tilde{x}_1^2- x_2\tilde{x}_1^2- x_1^2x_2\tilde{x}_1 + x_2\tilde{x}_1
\rangle$}
\end{itemize}
\noindent
From the \Gr\ basis we  find that $| \mathcal{V}({\mathcal{I}}_C^3)| = 65$
(\cite{CGC-cd-alg-vdim}).
Since
$| \Delta_2| = 9$, we have $A_1 + A_2 = (65 - 9)/2 = 28$, in
accordance with known values.
\end{example}
We can easily get an explicit formula for a generic $A_t$, since $|{\mathcal V}(\mathcal{I}_C^{t})|=2 \sum_1^{t-1} A_i+|\Delta_k|$.
\begin{theorem}
\label{formula_At}
$$
A_t=\frac{|{\mathcal V}(\mathcal{I}_C^{t+1})| \setminus
|{\mathcal V}(\mathcal{I}_C^t)|}{2}
$$
\end{theorem}
We provide a last example.
\begin{example}
For the same code of Example \ref{distr_dist}, we find $A_2$.
\begin{itemize}
\item[]\small{$\mathcal{I}_C^2=\langle
\tilde{x}_1\tilde{x}_2 - x_1\tilde{x}_2 - x_2\tilde{x}_1 + x_1x_2, \tilde{x}_2^3 - \tilde{x}_2, x_2^3 - x_2, x_1^3 - x_1, x_1\tilde{x}_2^2 - x_1x_2^2, x_1x_2^2\tilde{x}_2 - x_1\tilde{x}_2   , \tilde{x}_1^2 - x_2^2\tilde{x}_1 - \tilde{x}_1 + x_1x_2^2 - x_1^2 + x_1, x_1x_2^2\tilde{x}_1 - x_1^2\tilde{x}_1 + x_1\tilde{x}_1 - x_1^2x_2^2 - x_1^2 + x_1, x_1^2x_2\tilde{x}_1 + x_1x_2\tilde{x}_1 - x_1^2x_2 - x_1x_2
\rangle$}
\end{itemize}
\noindent
From the \Gr\ basis we obtain that $| \mathcal{V}({\mathcal{I}}_C^2)| = 25$. 
Applying Theorem {\ref{formula_At}},
 we can thus conclude  $A_2  = (65 - 25)/2=20$.
\end{example}

%


\section{Complexity considerations}
\label{complexity2}
The complexity of our method (Corollary \ref{mega}) is basically the complexity of the 
computation of
 a \GR\ basis for  ${\mathcal I}_C^t$, $2\leq t \leq d+1$, which is the core of the algorithm. The ideal requiring a bigger computational effort is obviously
$\mathcal{I}_C^t$ with the largest $t$ that we have to compute. 
 From now on, we try to  consider  the ideal 
$\mathcal{I}_C^t$ where
the ${\sf f}$ polynomials are as generic as possible.

 To compute the distance, we do not need the exact structure of the \Gr\ basis, but we just want to know if ${\mathcal V}
({\mathcal I}_C^t)$ is the diagonal.
We experimentally found that, if 
${\mathcal V}({\mathcal I}_C^t)$ is the diagonal, 
then the  basis obtained by auto-reductions is the diagonal as well.
Thus,  in our case  the complexity of the \Gr\ computation 
seems to be equivalent\footnote{They give different results, except
  when they output the diagonal.} to that of the computation of an
inter-reduced  basis  from the proposed generator set for $\mathcal{I}_C^t$.
 Hence, we focus on the complexity estimation for the relevant inter-reductions. %
\\
\\
\indent
We want to estimate the complexity to get an inter-reduced basis, starting  from the input 
basis of generators for ${\mathcal I}_C^t$. 
We will use some experimental results.
Let $L$ be a set of polynomials that we want to inter-reduce. Let $N=|L|$. A single reduction involves an element $f$ of $L$ that is reduced w.r.t. $L \setminus \{
f \}$. The output is either $0$, in which case $L$ becomes $L \setminus \{ f \}$ (and $N$ becomes $N-1$), or $f$ itself, if it cannot be reduced, or a new polynomial $f'$, which is the remainder and which replaces $f$ in $L$.
So, in the worst case, any reduction requires a division by $N-1$ polynomials and generates a new polynomial $f'$. Note that the $L$ size cannot increase. The starting step requires then, in the worst case, $\frac{N(N-1)}{2}$ divisions.
Let $r$  be the number of new polynomials which are obtained by reductions, both at the starting step and at any other subsequent step. Since $N$ does not grow, any of these $r$ new polynomials has to be divided by no more than $N-1$ polynomials.
In conclusion, the worst case estimate gives
 $
\frac{N(N-1)}{2}+rN
$ 
 divisions, and $2^n n$ is clearly the (worst case) cost of any division.\\
In our case, $N=\binom{n}{k}=\binom{2k}{k}$. By induction, we can
derive the  following lemma:
\begin{lemma}\label{ilaria}
 Let  $n,k\in\mathbb{N}$, with $1\leq k\leq n$, and $\alpha\in \mathbb{Q}$. If $s(\alpha)$ is such that $1+2^{\alpha}\leq 2^{s(\alpha)}$, then: 
$$
{n \choose k} 2^{\alpha k}\leq 2^{s(\alpha) n}\,.
$$
\end{lemma}
\noindent
We denote by $N$ the number of polynomials 
 in $$ S_{gen}=
  \{ x_i^q-x_i,\tilde{x_i}^q-\tilde{x_i} \mid 1\leq i\leq k \} \cup
  \{ {\sf m} (L_{n,k,t} ) \mid 
      {\sf m} \in\mathcal{M}_{n,t} \} \,.
$$
We treat the hard case, i.e. when $n=2k$ and $d \simeq k$.
 Despite our worst case formula, our computations suggest that it is possible to assume:
\begin{itemize}
\item $N$ initial divisions instead of $\frac{N(N-1)}{2}$,
\item $r= \sqrt{2^{2k}}=2^k$  ($2^{2k}$ is the number of all possible monomials),
\item $2^k n$ instead of $2^n n$, the average cost of each division,
\item $\sqrt{N}$ the average number of divisions for any new polynomial $f'$.
\end{itemize}
Thus, the computational cost {\sf C} to produce an inter-reduced basis from $S_{gen}$ 
can be estimated in
\begin{align*}
{\sf C}=\left( \binom{2k}{k}+ 2^k \sqrt{\binom{2k}{k}} \right) 2^k n &=\left[ \binom{2k}{k} 2^k
+2^{\frac{3}{2}k} \sqrt{\binom{2k}{k}2^k} \right]2k\\
&=\left[ 2^{3k}+2^{\frac{3}{2}k} 2^{\frac{3}{2}k} \right] 2k=(2^{3k}+2^{3k})2k = 2^{3k+1}2k,
\end{align*}
where we applied Lemma \ref{ilaria} with $\alpha=1$ and  $s(\alpha)=1.5$\footnote{actually $s(\alpha)$ is slightly larger.}.
We would like to compare our estimates with our numerical results. In order to do that we highlight the asymptotic exponential behaviour in $k$.
\begin{definition}
 Let $f,g:{\mathbb R} \mapsto {\mathbb R}$ s.t. $f(m)>0,\,g(m)>0$ for $m \geq 1$.
We say that $f \simeq g $ if and only if there exist $m_1, m_2, \alpha_1, \alpha_2 \, \geq 1$ such that
$$
  f(m) \leq g(m)m^{\alpha_1} \quad m \geq m_1,\,\qquad
  g(m)\leq f(m)m^{\alpha_2} \quad m \geq m_2.
$$
\end{definition}
It is easy to see that
  $\simeq$ is an equivalence relation and, given  
   $f \simeq 2^{\alpha n}$  and
$g \simeq 2^{\beta n}$, for $\alpha,\,\beta \geq 1, \, \alpha,\beta
\in \RR$, then  $f\simeq g$ if and only if $\alpha = \beta$. Moreover $\alpha > \beta$ implies $\lim_{n\to \infty} \frac{g}{f}=0$.

With this notation, our previous estimates for  ${\sf C}$ can be written as ${\sf
  C}\simeq 2^{3k}$.\\
\indent
We tested our problem using different computer algebra systems (Magma
2.10.13 \cite{CGC-MAGMA}, Polybori 0.3.1  \cite{CGC-poli},
Singular 2.0.6, Singular 3.0.4 \cite{CGC-GPS07}). 
For any system the time needed has a behavior of kind $2^{\alpha k}$
, with $\alpha$ depending on the system.
In particular, we report the graph where $x= k$ and $y= \log(\frac{time\, in\, k}{time\, in\, k-1})$, so that the $y$ values represent the expected exponent.
\begin{figure}
    \centering
     {\includegraphics[width=8cm,height=6cm]{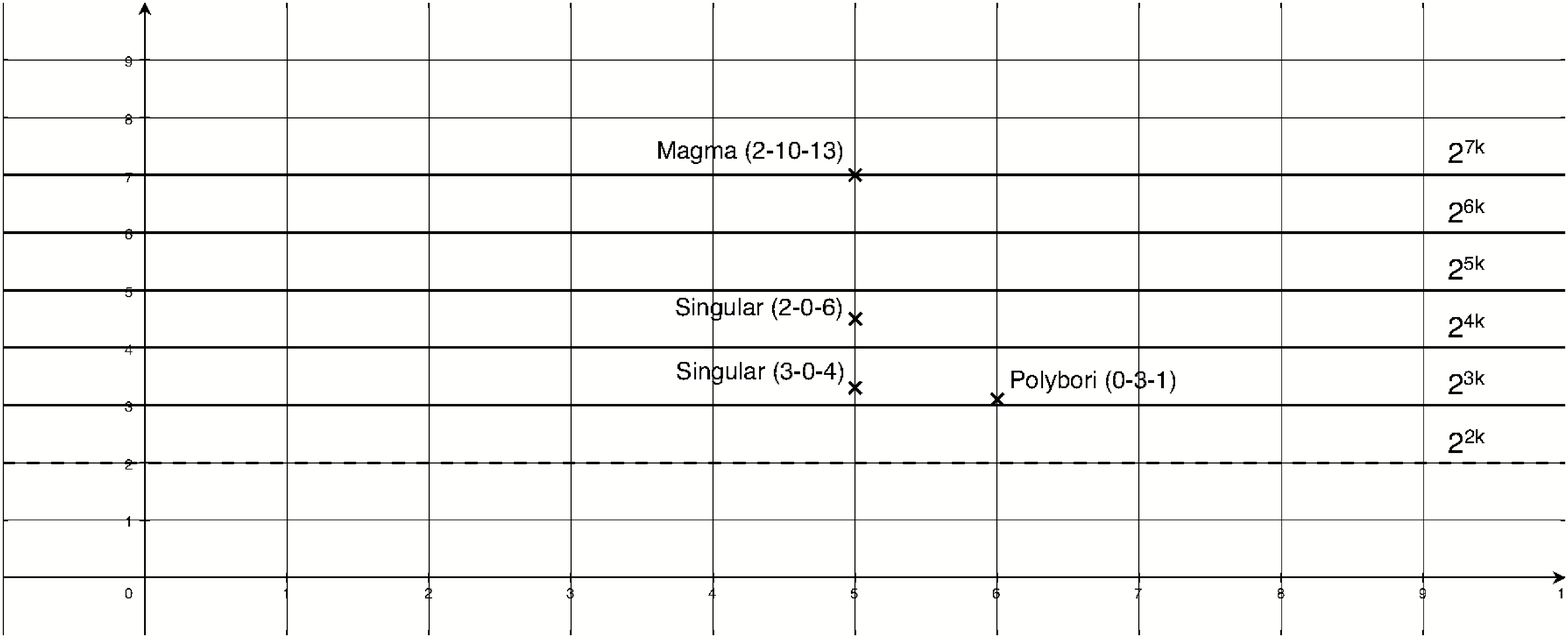}}
    \end{figure}
%
 This suggests us the following values for the computational costs:
 $2^{7k}$ for Magma 2.10.13, 
 $2^{4.5k} $ for Singular 2.0.6,
$ 2^{3.5k} $ for Singular 3.0.4,
 $2^{3k}$ for Polybori  0.3.1.
We note that a brute-force check of the distance has an asymptotic
behaviour like $2^{2k}$. 
Unfortunately, it is impossible to find an algorithm\footnote{The
  situation in the linear case is totally different. The problem of
  finding the distance is NP-complete (\cite{CGC-cd-art-vardy97}), so no sub-exponential algorithm
is known, but it might exist.} that computes the closest pairs
of a code with a lower complexity, since
$2^{2k}$ is indeed the complexity of the problem
(see Section \ref{pair}).
admittedly,
$2^{2k}$ looks   better than our estimate $2^{3k}$, but  two comments
are in order:
 \begin{itemize} 
\item If we examine the drastic improvement obtained by the evolution of computer algebra systems (from $2^{7k}$ to $ 2^{3k}$ in few years), we can reasonably assume that our estimates are still pessimistic and that further improvements in software development will allow our method to run like $2^{\alpha k}$, with $2< \alpha < 3$, possibly close to $\alpha =2$.
\item The brute-force check can output the distance only if the  
  code is given. When  a family of codes is given, that is, the
  $\{\rm{f}_i\}$ depend on some parameters $\{\lambda_j\}_{j\in J}$,
  the check is inapplicable, but
  our approach may be able to give general results for the family,
  such as a lower bound on the distance independent of the
  $\lambda_j$'s or even a ``a priori'' bound dependent on the  $\lambda_j$'s. 
\end{itemize}

\subsection{Complexity of the closest-pair problems}
\label{pair}
The determination of the minimum distance of a non-linear code is an instance
of the more general ``closest pair problem'', which is studied for
general metric spaces and, more deeply, for the Euclidean spaces $\RR^n$.
Several related problems are studied in this context, such as the
``nearest neighbour problem'' (which is the ``decoding problem'' in 
coding theory). An excellent reference is \cite{CGC-alg-book-prepshamos85},
to which we implicitly refer when we do not give explicit definitions
or quotations.

From now on, let $S$ be the number of points we are considering.
In the breakthrough 1975 paper \cite{CGC-alg-art-shamos75} the complexity
of many of these problems was established for the planar Euclidean case 
(i.e., for $\RR^2$).
In particular, it was shown that it is possible to solve the closest-pair
problem in $S \log(S)$ steps, by a clever application of the Voronoi diagrams.
And it was also shown that $S \log(S)$ matches exactly the complexity
of the problem. Since $S^2$ is the complexity of the 
``naive approach\footnote{this is the way the ``brute-force check'' is called in computational geometry.}'',
one might think that the naive approach could be beaten also in other metric
spaces.
Later, more refined algorithms have appeared that solve 
the closest-pair problem over $\RR^n$ with a claimed $S\log(S)$ complexity.
However, these algorithms' complexity is actually $S\log(S)f(n)$ for
some function $f:\NN \rightarrow \NN$, that is,
their complexity is computed by considering the space dimension {\em{fixed}}.
The fact here is that computational geometers are interested in cases
when the number of points is much larger than the space dimension.
In the Hamming space $(\FF_2)^n$ it is crucial to write explicitly also the
dependence on $n$, since clearly $S\leq 2^n$.
A deep analysis of their proofs shows that $f$ is exponential
in $n$ (e.g., about $4^n$ in Suri's divide-and-conquer algorithm
         \cite{CGC-alg-prep-suri}).
Translated into Hamming space language, this means that the naive algorithm
performs no worse than the others.
In fact, even if we allow for $f=2^n$, we will get an overall complexity
of $S\log(S)2^n=2^k \cdot k \cdot 2^n\sim 2^{3k}$ ($2^{2k}$ is the complexity 
of the naive approach). 
Of course, we would expect that the dependence on the dimension 
could be improved.
\begin{remark}
The instance in $(\FF_2)^n$ of the closest-pair problem is actually 
different from the ``distance computation problem'', since the latter
needs only to output the value of the distance and not the closest pair.
However, both our \Gr\ basis algorithm and the naive approach do output
the closest pair/s, and so we will consider the former from now on.
\end{remark}

When not dealing with Euclidean spaces (or metric spaces embeddable
with an isometry into $\RR^n$), the standard model in the literature for these
problems is the ``black box with distance oracle''. Basically it means
that the complexity is computed in terms of the number of distance calculations
which are necessary (a ``call to the distance oracle'' means a distance
calculation). As an example, it is possible to adapt the proof of 
Theorem 2.3 in \cite{CGC-alg-art-krauthgamee05} to prove the following
\begin{theorem} \label{Dd}
The (worst-case) complexity of decoding non-linear codes in $(\FF_2)^n$ is 
$\Omega(S)$.
\end{theorem}
Before we prove it, we need to explain the meaning of a ``proof''.
We are assuming that there is an algorithm $\mathcal{A}$,
which accepts as input both a non-linear code $X$ of size $S$ and a query $q$,
i.e. a point for which we  must find the point/s in $X$ closest to $q$.
Algorithm $\mathcal{A}$ is the best possible algorithm.
It calls the distance oracle many times and tries to use the information
on the computed distances to avoid making other distance computations.
What ${\mathcal{A}}$ can use is the triangle inequality:
\begin{equation} \label{triangle}
\dd(x,y) \leq \dd(x,z) + \dd(y,z), \quad \dd(x,z) \geq \dd(x,y) - \dd(y,z)
\end{equation}
since the algorithm is searching a distance minimum, it will use
the right-hand version of (\ref{triangle}). To prove Theorem \ref{Dd}
we need to exhibit, for any $n$ sufficiently large, a code of length $n$
and a query $q$ such that ${\mathcal{A}}$ is forced to perform $\Omega(S)$
distance computations.
\begin{proof}
Let $\ell=\lfloor\frac{n}{2}\rfloor$.
Let $X_n$ be the code in $(\FF_2)^n$ for $n\geq 4$ with query
$q_n$
$$ 
   X_n = \{ c\in (\FF_2)^n \mid \ww(c)=\ell \} \cup \{P=(1,0,\ldots,0)\}, \quad
   q_n = \{(0,\ldots,0) \}.
$$
In other words, $q_n$ is the zero vector and $X_n$ is the sphere centered in 
$q_n$ of radius $\ell$ plus a point $P$ at distance $1$ from $q_n$.
The crux of the proof here is that ${\mathcal{A}}$ will always get
$\ell$ from its distance computations, except for $\dd(q_n,P)$, which
will give $1$. So, it does not matter how smart ${\mathcal{A}}$ is, it
will {\em not able} to use any of its former distance computations.
Therefore, in the worst case, 
${\mathcal{A}}$ has to try all $\dd(q_n,x)$ for $x\in X$.
\end{proof}
Note that the size of $X_n$ in the above proof grows
with $n$ (actually $|X_n|>2^{n/2}$). 
Without this property, we could take $X_n$ with two points and 
then of course any algorithm will need to perform $1=S-1=\Omega(S)$ 
distance computations.

The previous digression is important in our opinion to put into context
the problem and understand the proof of our last result in this section,
which is the following theorem.
\begin{theorem}\label{paircl}
The (worst-case) complexity of computing the closest codeword pairs of 
a binary non-linear code is $\Omega(S^2)$.
\end{theorem}
\begin{proof}
For any subset $Y$ of $(\FF_2)^n$ we denote by $D$ its diameter,
that is, the maximum distance between two points in $Y$, and by $d$ its
(Hamming) distance.
Let ${\mathcal{B}}$ be any algorithm having as input a non-linear code
and returning a pair of closest codewords (we can think of
${\mathcal{B}}$ as the best possible).
To prove our claim we need to show that for any sufficiently large $n$ we can
find a code $X_n$ such that ${\mathcal{B}}$ cannot use any distance
computations already performed in order to discard other distance computations.

We consider $n\geq 10$. We take $X_n$ as any maximal subset of $(\FF_2)^n$
such that its {\em aspect ratio} $\frac{D}{d}$ is strictly lower
than $2$. The bound obtainable by ${\mathcal{B}}$ from any of its
former computations is at most
$$
  \dd(x,z) \geq \max {\dd(x,y)} - \min{\dd(y,z)} = D-d < 2d-d = d \,.
$$
This cannot give any help to ${\mathcal{B}}$ and so ${\mathcal{B}}$
is forced to compute also $\dd(x,z)$.
\end{proof}

We would like to give one final remark. The size of code $X_n$ 
in the previous proof clearly grows with $n$, but 
it will have only a few words.
If one is interested in building a larger code on which ${\mathcal{B}}$
needs to compute $\Omega(S^2)$, then one will need to use Ramsey-like
properties of the Hamming space (\cite{CGC-alg-art-bartal05}), similarly
to what is done in \cite{CGC-alg-art-krauthgamee05}.
To say it in a few words avoiding technicalities, it is possible
to find a subset $Y$ of $(\FF_2)^n$ with a size $2^{\Omega(n)}$ 
and such that its aspect ratio is limited to $2+\epsilon$, for some
small $\epsilon$.
Since  the number of minimum-distance word pairs and the number
of maximum-distance word pairs may be ``small'' compared to the number
of all word pairs, it follows that ${\mathcal{B}}$ may have to examine
in the worst case a number $\Omega(S^2)$ of word pairs before being
able to use (\ref{triangle}). However, we feel that a complete
proof for this claim is out of the scope of this paper and
 so we do not delve into it.

\section{Conclusions}
\label{conc}
The decoding performance of a distance-invariant code (e.g., a linear code) depends on its weight 
distribution,  although already the distance gives partial information on it.
For a generic non-linear code the performance depends on its distance
distribution (but the distance provides significant information).
\\
In this paper we have provided some \Gr\ techniques in order to compute the above-mentioned
code parameters in the systematic case (which is the most interesting).
We realize that no method can be faster than running  a specific C-programme optimized
for a given code and we are far from claiming that our techniques can compete with this
approach.
However, the optimized software has two drawbacks:
\begin{itemize}
\item the software writing and debugging can be long and complex, (while our methods are very easy to implement using a software package
for \Gr\ basis computations),
\item the software programme can be used only on a given code and cannot give general
results on a code family (while our methods could).
\end{itemize}
In our opinion this means that our methods can be of interest for a mathematician
investigating  theoretical code properties.

\section*{Acknowledgements}
\label{ack}
Part of these results can be found in
\cite{CGC-tesi2-guerrini} and \cite{CGC-cd-phdthesis-ele} and have been 
presented at MEGA2005 and Linz D1 2006 (\cite{CGC-cd-inbook-D1guerrini}), 
which was a workshop within the Special Semester 
on Groebner Bases, February--July 2006, organized by
RICAM, Austrian Academy of Sciences, and RISC, Johannes Kepler University, 
Linz, Austria.

The first two authors would like to thank their supervisor: the third author.

For their comments and suggestions, the authors heartily thank the
 anonymous referees and the following people: F. Caruso, 
 P. Fitzpatrick, 
P. Gianni, R. Krauthgamer, T. Mora, I. Simonetti and C. Traverso.

The authors would also like to thank the team at the computational centre MEDICIS
(http://www.medicis.polytechnique.fr/).

This work has been partially supported by STMicroelectronics contract
``Complexity issues in algebraic Coding Theory and Cryptography''.

\newcommand{\etalchar}[1]{$^{#1}$}
\providecommand{\bysame}{\leavevmode\hbox to3em{\hrulefill}\thinspace}
\providecommand{\MR}{\relax\ifhmode\unskip\space\fi MR }
\providecommand{\MRhref}[2]{%
  \href{http://www.ams.org/mathscinet-getitem?mr=#1}{#2}
}
\providecommand{\href}[2]{#2}

\end{document}